\documentclass[twoside,a4paper,12pt]{article}
\usepackage[latin1]{inputenc} 
 \usepackage[T1]{fontenc}      
\usepackage{geometry}         
\RequirePackage{textcomp}
\RequirePackage{amsthm}
\RequirePackage{amsmath}
\RequirePackage{bm}
\RequirePackage{amssymb}
\RequirePackage{upref}
\RequirePackage{amsthm}
\RequirePackage{enumerate}
\RequirePackage{pb-diagram}
\RequirePackage{amsfonts}
\RequirePackage[mathscr]{eucal}
\RequirePackage{verbatim}
\RequirePackage{xr}
\newcommand{\al}{\alpha}
\newcommand{\bet}{\beta}

\newcommand{\e}{\epsilon}

\newcommand{\h}{\eta}

\newcommand{\bo}{\mathbf}

\newcommand{\Lb}{\underline{L}}

\newcommand{\demi}{\frac{1}{2}}
\newcommand{\ls}{\leqslant}
\newcommand{\rs}{\geqslant}
\newcommand{\hs}{\hspace{5mm}}

\newcommand{\abs}[1]{\lvert#1\rvert}

\newcommand{\norm}[1]{\lVert#1\rVert}

\newcommand{\normf}[1]{\left\|#1\right\|}

\usepackage{xspace}

\newcommand{\Dep}{\mathcal{D}^+(\Omega)}

\newcommand{\sstrike}[1]{\ensuremath{\not\!\text{#1}}\xspace}
\newcommand{\nablab}{\sstrike{$\nabla$}}

\newtheorem{thm}{Theorem}[subsection]
\newtheorem{lem}{Lemma}[subsection]

\newtheorem{defi}{Definition}[subsection]
\newtheorem{cor}{Corollary}[subsection]
\newtheorem{conj}{Conjecture}[subsection]
\begin{document}

\title{Local extendability of Einstein vacuum manifolds}

\author{David Parlongue}
\date{\today}
\maketitle
\begin{abstract} We revisit in this article results of Klainerman and Rodnianski on a geometric breakdown criterion for Einstein vacuum spacetimes. We take advantage of the use of a time-harmonic transversal gauge to give a localized version (in space and time) of this result.

\end{abstract}
\thispagestyle{empty}

\section{Introduction}

\hspace{5mm} The Cauchy problem for Einstein vacuum equations has been largely studied in the last fifty years since the first results on local well-posedness to the non-linear stability of the Minkowski spacetime and beyond low-regularity solutions, relations between general relativity and geometrization ... The question of formation of singularities is however quite open and remains one of the main challenge of classical general relativity. 

  From a mathematical point of view a geometric breakdown criterion involving only control on the $L^\infty$ norm of the second fundamental form and the horizontal derivative of the lapse has been proved for Einstein vacuum equations (EVE) in the constant mean curvature gauge (CMC) as well as in the asymptotically flat case (see \cite{krb}). These results have been extended to criteria with weaker conditions and also in the direction of Einstein equations coupled with a scalar field or Maxwell equations. 
  
  These results are however global in space and do not give information on local extendability of vacuum spacetimes. The goal of this article is to localize in space these criteria and apply these local theorems to prove a result on formation of singularities and local extendability properties of vacuum spacetimes. We first review some continuation results for EVE.
\section{Global breakdown criteria}

\hspace{5mm} We consider in this paper a $3+1$ dimensional Lorentzian manifold $\mathbf{(M,g)}$ satisfying the EVE :

\begin{equation}
\begin{aligned}
\mathbf{R_{\al \bet}(g)=0}
\end{aligned}
\end{equation}

We suppose that a part of the space-time denoted by  $\mathcal{M} \subset \mathbf{M}$ is globally hyperbolic with respect to a spatial hypersurface $\Sigma_0$ and foliated by the level hypersurfaces of a regular time function denoted by $t$. We suppose moreover that this time function is monotonically increasing towards future, with lapse function $n$ and second fundamental form $k$. We will use the following sign convention  :
\begin{eqnarray}
&& k(X,Y)=-\bf{g}(D_{X}T,Y) ,\\
&& n=(-\bf{g}(\bf{D}t,\bf{D}t))^{-1/2}
\end{eqnarray}

Throughout this paper $\bf{T}$, will denote the future unit normal to the leaf $\Sigma_{t}$ of the time foliation, $^{\bf (T)}{\pi}$ its deformation tensor and $\bf{D}$ the covariant derivative associated with  $\bf{g}$. Let $\Sigma_{0}$ be a slice of the time foliation which will be referred as the initial slice. The EVE can be seen as an initial value problem once an initial data set is given. An initial data set consists in a 3 dimensional Riemanian hypersurface $\Sigma_{0}$ together with a 3 dimensional Riemannian metric $g$ and a second fundamental form $k$. 

$(\Sigma_0,g,k)$ has moreover to satisfy the following constraint equations which take the following form in the maximal gauge regime :
\begin{eqnarray}
&& tr k = 0,\\
&& \nabla^{j}k_{i,j}=0,\\
&&
 R=|k|^2
\end{eqnarray}
where here $R$ stands for the Riemann curvature of the metric $g$ and $\nabla$ the covariant derivative associated to $g$.
The other equations coming from the reformulation of EVE are the following evolution equations :
\begin{eqnarray}
&&\partial_t g_{ij}=-2nk_{ij},\\
&&\partial_t k_{ij}=-\nabla_i\nabla_j n + n(R_{ij}-2k_{ia}k^{a}_j)
\end{eqnarray}

together with the lapse equation :

\begin{equation}
\begin{aligned}
\Delta n = |k|^2 n
\end{aligned}
\end{equation}
with condition $n \rightarrow 1$ at spatial infinity on $\Sigma$.

The continuation result of  \cite{krb}  holds true for maximal foliations as well as for constant mean curvature foliations (CMC). In this case the leaves are compact and the mean curvature can be taken as being the time function. We will thus consider here the case of maximal foliation, the one with CMC foliation being simpler due to the compacity of the leaves. 

 $\Sigma_0$ is said to be asymptotically flat if the complement of a compact set $K \subset \Sigma_0$ is diffeomorphic to the complement of a $3$-sphere and that there exists a system of coordinate in which :

\begin{eqnarray}
g_{ij}&=&(1-\frac{M}{r})\delta_{ij} +o_4(r^{-3/2}) \\
k_{ij}&=& o_3(r^{-5/2})
\end{eqnarray}

\vspace{3mm}

{\bf [Geometric assumption \bf (G)]} The initial surface $\Sigma_0$ is said to satisfy {\bf (G)} if and only if there exists a converging of $\Sigma_0$ by a finite number of charts $U$ such that for any fixed chart, the induced metric $g$ verifies :

$$C^{-1} \abs{\xi}^2 \ls g_{ij}(x)\xi_i \xi_j \ls C\abs{\xi}^2 , \hs \forall x \in U$$

with $C$ a fixed number.
\vspace{3mm}
 
The system of equations (4) to (9) is a determined system of equation. We are naturally led to ask the question of well-posedness for this system. The first result of this theory is the local well-posedness result of Choquet-Bruhat  \cite{cb}, here as stated in \cite{ck}:

\begin{thm} (Local existence theorem) Let $(\Sigma_0,g_0,k_0)$ be an initial data set verifying the following conditions :

1) $(\Sigma_0,g_0)$ is a complete Riemannian manifold diffeormorphic to $\mathbb{R}^3$.

2)The ispoperimetric constant of $(\Sigma_0,g_0)$ is finite.

3) $Ric(g_0) \in {H}_{2,1}(\Sigma_0,g_0)$ and $k_0 \in {H}_{3,1}(\Sigma_0,g_0)$.
 
4) The initial slice satisfies the constraint equations.

Then there exists a unique, local-in-time smooth development, foliated by a normal, maximal time foliation $t$ with range in some $\left[0,T\right]$ and with $t=0$ corresponding to the initial slice $\Sigma$. Moreover :
 
 a) $g(t)-g \in \mathcal{C}^1([0,T],H_{4,1})$
 
 b) $k(t)\in \mathcal{C}^0([0,T],H_{3,1})$
 
 c) $Ric(g)(t)\in \mathcal{C}^0([0,T],{H}_{2,1})$
 \end{thm}
where for a given tensorfield $h$ on $\Sigma$, $\|h\|_{H_{n,s}}$ denotes the norm :
$$\|h\|_{H_{n,s}}= \huge( \sum_{i=0}^n\int_\sigma (1+d_0^2)^{s+i}\abs{\nabla^i h}^2 \huge)^\demi$$
\vspace{5 mm}

where $d_0$ stands for the geodesic distance to a basepoint $O$.

 The natural question for such a local result is the minimum regularity of the initial data required for having well-posedness. The scaling property of Einstein vacuum equation could lead to think that the critical Sobolev exponent for $g$ should be $3/2$. In wave coordinates, the local existence regularity can be extended to $5/2 + \e $ Sobolev regularity. An important improvement led to the proof for $2 + \e $ (see \cite{kr2}). The so-called $L^2$-conjecture asserts that the optimal regularity for the initial data set is $Ric(g)\in L^2(\Sigma_0)$ and $\nabla k \in L^2(\Sigma_0)$.

\vspace{2 mm}
We are now ready to state the main result of  \cite{krb} :

\begin{thm}
Let  $\mathbf{(M,g)}$ be a globally hyperbolic development of an asymptotically flat initial data set $(\Sigma_{0},g_0,k_0)$ satisfying the assumptions of the local in time theorem and globally foliated by a normal, maximal foliation given by the level sets of a smooth time function t such that $\Sigma_{0}$ corresponds to $t=0$. Suppose moreover that $\Sigma_{0}$ satisfies the geometric condition {\bf (G)}. 
We suppose that :
\begin{eqnarray}
&a)& \int_{\Sigma_{0}} \abs{{ \bf R}}^2(x)d\mu(x) \ls \Delta_0\\
&b)& \normf{n^{-1}}_{L^\infty([0,t_1[,L^\infty(\Sigma_{t}))} \ls\Delta_1\\
&c)& \normf{k}_{L^\infty([0,t_1[,L^\infty(\Sigma_{t}))} + \normf{\nabla(log(n))}_{L^\infty([0,t_1[,L^\infty(\Sigma_{t}))} \ls \Delta_2
\end{eqnarray}
than the space-time together with the maximal foliation can be extended beyond time $t_1$.
\end{thm}

This theorem implies  the following breakdown criterion :

\begin{thm} [Breakdown Criterion]
With the notations of the previous theorem, the first time $T*$ with respect to the $t$-foliation of a breakdown is characterized by the condition :

\begin{equation}
\limsup_{t \to T*} {\big( \normf{k}_{L^\infty(\Sigma_{t})}+ \normf{\nabla(log(n))}_{L^\infty(\Sigma_{t})}\big)} =+ \infty
\end{equation}
\end{thm}

The pointwise condition can be relaxed to obtain the following theorem (see \cite{dp})

\begin{thm} {\bf [Integral Breakdown Criterion]}
We keep here the notations of the previous theorem. We suppose that there exists $\Delta_0, \Delta_1$ and $\Delta_2$ such that : 
\begin{eqnarray}
&a)& \int_{\Sigma_{0}} \abs{{\bf R}}^2(x)d\mu(x) \ls \Delta_0\\
&b)& \normf{n^{-1}}_{L^\infty([0,t_1[,L^\infty(\Sigma_{t}))} \ls\Delta_1\\
&c)& \int_0^{t_1}(\normf{k}_{L^\infty(\Sigma_{t})}+ \normf{\nabla(log(n))}_{L^\infty(\Sigma_{t})})^2 n dt \ls \Delta_2
 \end{eqnarray}

than the space-time together with the maximal foliation can be extended beyond time $t_1$.
\end{thm}

This theorem implies of course the corresponding breakdown criterion. 
In both cases, the breakdown criteria are however (spatially) global. However, the question of physical relevance related to continuation of solutions of Einstein equations is the question of local extendability. Obtaining information on local extendability of spacetimes requires to localize in space these breakdown criteria. 

Let us recall some definitions relevant to cosmic censorship conjectures.
\begin{defi}
A spacetime {\bf (M,g)} is said locally inextendible if there exists no open subset $\mathcal{U} \subset {\bf M}$ with non-compact closure in ${\bf M}$ such that there exists an isometric imbedding $\phi : (\mathcal{U}, {\bf g}_{|\mathcal{U}}) \rightarrow (\mathcal{U}', {\bf g'})$ where $\phi(\mathcal{U})$ has compact closure in $\mathcal{U}'$.
\end{defi}

\begin{conj}[Strong Censorship Conjecture]
The maximal globally hyperbolic vacuum extension (MGHVE) of a "generic" initial data set $(\Sigma,g,k)$ where $\Sigma$ is either compact or asymptotically flat is locally inextendible as a $\mathcal{C}^{1,1}$ Lorentzian manifold.
\end{conj}

The difficulty throughout this paper will be that we will have to deal with slices $S_t$ which are non-complete Riemmanian manifold. Among the consequences, the solution constructed from data $(S_t,g,k)$ will typically have a temporal extent tending to zero near the boundary of $S_t$. From another hand, some techniques based on elliptic estimates on the slices used in the proof of global criteria with $\Sigma$ compact or asymptotically flat have to replaced by purely hyperbolic estimates. This will motivate the use of a transversal time-harmonic gauge.

\section{Local breakdown of EVE :}
\subsection{The transversal time-harmonic gauge :}

\hspace{5mm}The main goal of this paper will be to localize the breakdown criteria of part 2. Among the different choices of gauge possible, some of them have been largely used because of their physical significance (asympotically flat spacetimes and maximal gauge, CMC gauge) or because of their mathematical properties (the wave gauge enables for instance to reduce the principal symbol of the Ricci tensor to $-\frac{1}{2}\square_{\bo g}{\bo g}_{\al \bet}$). We will however use here another choice of gauge well-suited for spatially localized theorems, a transversal time-harmonic gauge. The reason why this type of gauge condition is well-suited for local-in-space applications is that the reduced system that we will obtain is hyperbolic instead of being a mixed hyperbolic-elliptic system as in the case of the CMC or AF gauge where we have to solve an elliptic system for the lapse. In comparison with the harmonic gauge, we will also take advantage of the transversality.

Fixing a gauge requires to fix four quantities while working on $\Sigma \times \mathbb{R}$ with transported coordinates, where $\Sigma$ stands for a 3-dimensional Riemannian manifold. The choice made here is having a vanishing shift or equivalently a choice of time-lines orthogonal to the space-sections and an harmonic time index. That is if denote by $(x^0,x^2,x^3,x^4)$ by a local set of coordinates :

\begin{equation}
\square_{\bo g} t = -{\bo g}^{\al \bet}\Gamma^0_{\al \bet}=\frac{1}{2n^2} (\frac{\partial}{\partial x^0}(log(g))-  2\frac{\partial}{\partial x^0}(log(n))=0
\end{equation}

where $n$ stands for the lapse and $g$ for the the determinant of the three-dimensional Riemannian metric on the spatial-sections.
We see immediately that the harmonicity condition on the time-index leads to the following equality on the whole four-dimensional manifold :

\begin{equation}
n=\sqrt{g}f
\end{equation}

where $f>0$ stands for a scalar density of order one given on $\Sigma$. Equivalently, we can give ourself a Riemannian metric $e$ on $\Sigma_0$ and set $f=(det(e))^{-1/2}$.

We set moreover :

\begin{equation}
P^{ij}=k^{ij}-trk g^{ij}
\end{equation}

The reduction of the Cauchy problem for Einstein equations in this particular gauge was performed in \cite{cb} :

\begin{eqnarray}
\frac{\partial}{\partial x^0} g^{ij}&=& n ( 2 P^{ij} -g^{ij}P) \\
 \square_{\bo g} P^{ij}&=&nM^{ij}\\
 \frac{\partial}{\partial x^0}(log(n)) & =&\demi \frac{\partial}{\partial x^0}(log(g))=-ntr(k) \\ 
 \nonumber
\end{eqnarray}
with a set of data $S=(g_0,P_0, \frac{\partial}{\partial x^0}P_0,e)$. Where $M$ stands for lower-order terms (of order at most 2 in $g$ or $n$ and one in $P$). This is a quasi-diagonal hyperbolic system denoted by $\mathcal{S}$ for which standard techniques apply.

\vspace{2mm}
Recalling that :
\begin{equation} 
\Delta n = \abs{k}^2n -  \frac{\partial}{\partial x^0} tr(k)
\end{equation}

and differentiating (24), we get :

\begin{eqnarray}
\frac{\partial^2}{\partial x^{0\hspace{1mm}2}} n &=& -2n(\frac{\partial}{\partial x^0}n)trk-n^2\frac{\partial}{\partial x^0}trk \\
&=& 2n^3 (trk)^2 -n^3 \abs{k}^2+n^2\Delta n
\end{eqnarray}
or
\begin{eqnarray}
\square_{\bf g}n&=&n(\abs{k}^2-2(trk)^2)+ (^{(3)}{\Gamma}^{i} - ^{(4)}{\Gamma}^{i})\nabla_i n   \nonumber \\
&=&n(\abs{k}^2-2(trk)^2+  \nabla^i (logn)  \nabla_i (logn))
\end{eqnarray}

Note that for a general transversal gauge, we have :
\begin{equation*}
\square_{\bf g}n=n(\abs{k}^2-2(trk)^2 +  \nabla^i (logn)  \nabla_i (logn)) - \frac{\partial}{\partial x^0}(n\Gamma^0)-2n^3(\Gamma^0)^2+4n^2trk\Gamma^0
\end{equation*}

The time-harmonic gauge plays an import role in the simplification leading to a fully hyperbolical structure for ${^{(\bo T)}{\pi}}_{\mu \nu}$.

Now if we denote by ${\bo P}_{\al \bet}= {\bo g} + {\bo T}_{\al} {\bo T}_{\bet}$ the projection operator on $\Sigma_t$ and if we denote without change of notation $k_{\al \bet}=-\demi {\bo P}_{\al}^{\nu} {\bo P}_{\bet}^{\mu} {^{(\bo T)}{\pi}}_{\mu \nu}$

It has been proved in \cite{cb} that :

\begin{eqnarray*}
\large( \frac{1}{n^2}\frac{\partial^2}{\partial x^{0\hspace{1mm}2}}-\Delta_{ g} \large) (nk)_{ij}&=& -3nk_{h(i}R_{j)}^h+2R_{i\hspace{2mm}j}^{\hspace{2mm}h\hspace{2mm}m}nk_{hm}+2\nabla_{(i}log(n)\nabla_{j)}(ntrk)+4nk\nabla_i\nabla_jlog(n) \nonumber \\
&&+nk_{h(i}\nabla_{j)}\nabla^hlog(n)+\nabla^hlog(n)\nabla_hnk_{ij}+3nk\nabla_i log(n)\nabla_j log(n)\\
&&+4nk^3_{ij}
\end{eqnarray*}

\hspace{3mm}

Schematically, we can write :
 \begin{equation}
\square_{\bo g} k_{ij} \approx ^{\bo (T)}{\pi}^3 +{\bo R}\cdot ^{\bo (T)}{\pi} +^{\bo (T)}{\pi}\cdot {\bo D}^{\bo (T)}{\pi}
 \end{equation}
 
 where we have used that $R\approx \bo R+ k^2$. In fact, we remark that $^{\bo (T)}\pi$ satisfies a wave equation. Indeed using (28), ${\bo T}=n {\bo D}t$ and :
 
 \begin{equation}
 \square_{\bo g} {\bo D}_{\al} {\bo D}_{\bet} t = - {\bo R}_{\al \hspace{2mm} \bet }^{\hspace{2mm} \mu\hspace{3mm} \nu} {\bo D}_{\mu} {\bo D}_{\nu} t 
 \end{equation}
 
 we get :
 
 \begin{equation}
\square_{\bo g} {^{\bo (T)}{\pi}} \approx  ^{\bo (T)}{\pi}^3 +{\bo R}\cdot ^{\bo (T)}{\pi} +^{\bo (T)}{\pi}\cdot {\bo D} ^{\bo (T)}{\pi}
\end{equation}

From now on, in this part we will suppose that the two following properties are satisfied :
\begin{itemize}
\item the Riemannian manifolds $\Sigma_t$ satisfy uniformly {\bf (G)}
\item there exists $C>1$, s.t. $C^{-1}\leq n \leq C$
\end{itemize}

These properties will be satisfied as a consequence of (24) in our applications.

The idea underlying the rest of this part is that the cubic wave equation is locally well-posed in dimension 3+1 for initial data in $H^\demi(\mathbb{R}^3)$. 
Let us now consider a generic equation of the form :

\begin{equation}
\square_m \phi = P(\phi, {\bo D}\phi)
\end{equation}
where $\phi$ stands for a scalar function $\phi : \mathbb{R}^3\rightarrow \mathbb{R}$, and 
\begin{equation*}
P(X,Y)=a_{00}+a_{10}X+a_{01}Y+a_{20}X^2+a_{3,0}X^3+a_{11}XY
\end{equation*}
where the $a_{ij}$ are supposed to be given in $C_0^{\infty}(\mathbb{R}^{1+3})$, we suppose also that  a set of initial data $(\phi_{t=0}, \partial_t \phi_{t=0})=(\phi_0,\phi_1)$ is given in $C_0^{\infty}(\mathbb{R}^{1+3})\times C_0^{\infty}(\mathbb{R}^{1+3})$. Then by standard techniques, the Cauchy problem (32) with initial data 
$(\phi_0,\phi_1)$ is locally well-posed, and the unique solution $\phi$ can be extended to any slab $[0,T]\times \mathbb{R}^3$ as long as :

\begin{equation}
\int_0^T\norm{\phi}_{L^\infty}dt'<\infty
\end{equation}

Note that it is crucial that we do not allow terms of the form $a_{02}Y^2$, a continuation condition for such an equation would require a $L^1_tL^\infty_x$ control on $\partial \phi$ and not only on $\phi$. This is typically the case of the toy-model for Einstein equation given by $\square \phi= \partial \phi \cdot \partial \phi$. 

Having in mind this model equation, we can hope use to prove a continuation principle for the wave equation satisfied by $^{\bo (T)}{\pi}$, provided that we are able to deal with a curved spacetime with rough metric instead of a Minkowskian one. 

\vspace{3mm}

We consider from now on a regular subset of $\Sigma_t$, $\Omega_t$ and denote by $\mathcal{I}$ its past and $\Omega_s=\Sigma_s\cap\mathcal{I}$. Let $\mathcal{H}$ denotes its lateral boundary which is supposed to be Lipschitz and $\bo{L}$ a null generator. This type of energy estimates for Einstein equations localized in a domain has been used in \cite{CBlocal}.

\begin{thm}
 $^{\bo (T)}{\pi}\in L^1_tL^\infty(\Omega_t)$,  $^{\bo (T)}{\pi}\in L^2(\Omega_0)$ and  ${\bo D}^{\bo (T)}{\pi}\in L^2(\Omega_0)$ implies  $^{\bo (T)}{\pi}\in L^\infty_tL^4(\Omega_t)$ and ${\bo D}^{\bo (T)}{\pi}\in L^\infty_tL^2(\Omega_t)$.
\end{thm}
\begin{proof}
We introduce the energy-momentum tensor $O_{\al \bet}$ associated to the previous wave equation :

\begin{eqnarray*}
O_{\al \bet}&=& h^{\delta \delta'}h^{\gamma \gamma'} ... {\bf D}_{\al} ^{\bo (T)}{\pi}_{\delta \gamma ...}{\bf D}_{\bet} ^{\bo (T)}{\pi}_{\delta' \gamma' ...} \\
&&-\demi {\bf  g}_{\al \bet} {\bf  g}^{\mu \nu} h^{\delta \delta'}h^{\gamma \gamma'} ... {\bf D}_{\mu} ^{\bo (T)}{\pi}_{\delta \gamma ...}{\bf D}_{\nu} ^{\bo (T)}{\pi}_{\delta' \gamma' ...}
\end{eqnarray*}
$O$ satisfies the following relations :
$$
O({\bf T}, {\bf T}) = \demi \abs{^{\bo (T)}{\pi}}^2
$$
as well as the following positivity property, forall pair of fututre oriented timelike or null vectorfield $({\bf S_1},{\bf S_2})$, $O({\bf S_1},{\bf S_2}) \geq 0$.
 
Moreover, we have :

\begin{equation*}
\abs{{\bo D}^\al(O_{\al0})}\lesssim  (\abs{ ^{\bo (T)}{\pi}}^3\abs{{\bo D}^{\bo (T)}{\pi}} +\abs{{\bo R}}\abs{^{\bo (T)}{\pi}}\abs{{\bo D}^{\bo (T)}{\pi}} +\abs{^{\bo (T)}{\pi}}\abs{{\bo D}^{\bo (T)}{\pi}}^2)
\end{equation*}
Denote by $A_t=\int_{\Omega_t}\abs{{\bo D} {^{\bo (T)}{\pi}}}^2-\int_{\Omega_0}\abs{{\bo D} {^{\bo (T)}{\pi}}}^2$. As a consequence of Stokes formula :

\begin{eqnarray*} 
A_t + \int_{\mathcal{H}}O(\bo{T},\bo{L})d\mu_{\mathcal{H}} &\leq&\int_0^t\int_{\Omega_s}(\abs{ ^{\bo (T)}{\pi}}^3\abs{{\bo D}^{\bo (T)}{\pi}} +\abs{{\bo R}}\abs{^{\bo (T)}{\pi}}\abs{{\bo D}^{\bo (T)}{\pi}} +\abs{^{\bo (T)}{\pi}}\abs{{\bo D}^{\bo (T)}{\pi}}^2)d\mu_snds\\
&\leq& \int_0^t \big( n\normf{{^{\bo (T)}{\pi}}}_{L^\infty(\Omega_s)}\int_{\Omega_s}(\abs{ ^{\bo (T)}{\pi}}^2\abs{{\bo D}^{\bo (T)}{\pi}} +\abs{{\bo R}}\abs{{\bo D}^{\bo (T)}{\pi}} +\abs{{\bo D}^{\bo (T)}{\pi}}^2)\big)d\mu_sds\\
&\lesssim&\int_0^t \normf{{^{\bo (T)}{\pi}}}_{L^\infty(\Omega_s)} (\normf{{\bo R}}_{L^\infty_tL^2_x}^2 + \normf{{\bo D}^{\bo (T)}{\pi}}_{L^2(\Omega_s)}^2+ \normf{{^{\bo (T)}{\pi}}}^4)nds
\end{eqnarray*}
using Gronwall lemma and the positivity of the flux  $\int_{\mathcal{H}}O(\bo{T},\bo{L})d\mu_{\mathcal{H}}$ :
\begin{equation}
\normf{{\bo D}{^{\bo (T)}{\pi}}}_{L^\infty_tL^2_x(\Omega_t)}^2 \lesssim C + \normf{{^{\bo (T)}{\pi}}}_{L^\infty_tL^4_x(\Omega_t)}^4
\end{equation}
 but remarking that if $0\leq u\leq s \leq t$, $F^u_s(\Omega_s) \subset \Omega_u$, we have for $0\leq s\leq t$ (where $F_s^u: \Sigma_s \rightarrow \Sigma_u$ is obtained by following the integral lines of $-\partial_t$).
\begin{equation*}
 \normf{{^{\bo (T)}{\pi}}}_{L^4_x(\Omega_s)}^4\leq \normf{{^{\bo (T)}{\pi}}}_{L^4_x(\Omega_0)}^4 +\int_0^s \normf{{^{\bo (T)}{\pi}}}_{L^\infty(F^u_s(\Omega_s))}\big( \normf{{\bo D}{^{\bo (T)}{\pi}}}_{L^2(F^u_s(\Omega_s))}^2 +\normf{{^{\bo (T)}{\pi}}}_{L^4(F^u_s(\Omega_s))}^4  \big)
\end{equation*}
which implies :
\begin{eqnarray*}
 \normf{^{\bo (T)}{\pi}}_{L^\infty_uL^4_x(F^u_s(\Omega_s))}^4 &\leq& C +D\int_0^s \normf{^{\bo (T)}{\pi}}_{L^\infty(F^u_s(\Omega_s))}\normf{{\bo D}^{\bo (T)}{\pi}}_{L^2(F^u_s(\Omega_s))}^2 ndu \\
 &\leq& C+D\int_0^s \normf{^{\bo (T)}{\pi}}_{L^\infty(\Omega_u)}\normf{{\bo D}^{\bo (T)}{\pi}}_{L^2(\Omega_u)}^2 ndu
\end{eqnarray*}

thus :
 \begin{equation*}
 \normf{^{\bo (T)}{\pi}}_{L^4_x(\Omega_s)}^4 \leq C +D\int_0^s \normf{{^{\bo (T)}{\pi}}}_{L^\infty\Omega_s)}\normf{{\bo D}^{\bo (T)}{\pi}}_{L^2(\Omega_s)}^2 ndu
\end{equation*}

plugging in (34) gives $^{\bo (T)}{\pi}\in L^\infty_tL^4(\Omega_t)$ and ${\bo D}^{\bo (T)}{\pi}\in L^\infty_tL^2(\Omega_t)$. 
 
 \end{proof}
 
 Deriving wave equations for ${\bo D^{\bo k}}^{\bo (T)}{\pi}$, $\abs{k}\geq1$ and using the same techniques, we deduce :
 
 \begin{thm}
If for a $k\geq1$, $^{\bo (T)}{\pi}\in L^1_tL^\infty(\Omega_t)$,  $\bo{D^l}^{\bo (T)}{\pi}(t=0)\in L^2(\Omega_0)$ for $l\leq k$ and $\bo {D^l}{\bo R}\in L^\infty_tL^2(\Omega_t)$ for $l\leq (k-1)$  then :

 \begin{center}
$ \bo {D^l}^{\bo (T)}{\pi}\in L^\infty_tL^2(\Omega_t)$ for $l\leq k$.
\end{center}
  \end{thm}
  
 We also introduce the reduced flux of $^{\bo (T)}{\pi}$ as :
 
 \begin{equation}
\mathcal{F}(\pi,p,\delta)=\int_{\mathcal{N}^{-}(p,\delta)} \abs{\nablab ^{\bo (T)}{\pi}}^2 +  \abs{\nablab_{\bo{L}} ^{\bo (T)}{\pi}} ^2
 \end{equation}
 
where $\mathcal{N}^{-}(p,\delta)$ stands for the past null cone with vertex $p$ and $\delta$ the null past radius of injectivity at $p$. We denote by $\bo{L}$ the (normalized) null geodesic generator of $\mathcal{N}^{-}(p,\delta)$ that is the vectorfield satisfying $<{\bo L},{\bo L}>=0$, $\bo{D}_\bo{L} \bo{L} =0$ and  $<\bo{T},\bo{L}>(p)=1$. We denote by $s$ the corresponding affine parameter and  $S_s$ the level spheres. $\nablab$ denotes the restriction of $\bo{D}$ to $S_s$ and $\nablab_L$ the projection to $S_s$ of $\bo{D}_L$.

We also have that :
\begin{thm}
For all $p\in {\bo M}$, 
\begin{equation}
\mathcal{F}(\pi,p,\delta) \lesssim C(\norm{^{\bo (T)}{\pi}}_{L^1_tL^\infty(\Omega_t)},  \norm{^{\bo (T)}{\pi}}_{ L^2(\Omega_0)},\norm{{\bo D}^{\bo (T)}{\pi}}_{L^2(\Omega_0)}, \norm{{\bo R}}_{L^2(\Omega_0)})
\end{equation}
\end{thm}

\subsection{Solving the reduced system :}

\hspace{5mm}Once the hyperbolic reduction has been performed, by a simple domain of dependance  argument it's sufficient to consider that the initial data is supported in a fixed coordinate patch. Gauge conditions are propagated (see \cite{cb}) and considering an initial data set $(\Sigma_0,g_0,k_0)$ satisfying the constraint equations and such that $(g,P)$ satisfy the previous quasi-linear hyperbolic system $\mathcal{S}$ on $\Sigma_0 \times \mathbb{R}$, then the metric :

\begin{equation}
{\bf g}=-n^2dt^2+g_{ij}dx^idx^j 
\end{equation}

satisfies Einstein vacuum equations on this set. Moreover, let $\Omega$ be a given subset of $\Sigma$ and $\mathcal{D}(\Omega)$ its domain of dependance for ${\bf g}$, if $(g,P)$ satisfies $\mathcal{S}$ on $\mathcal{D}(\Omega)$, the metric {\bf g} satisfies Einstein vacuum equations on $\mathcal{D}(\Omega)$.
We will now give a version of the local in time and existence which is gauge invariant. Moreover we track the dependance of the temporal extent of the MGHVE on the geometry of the initial slice.

\begin{thm}{\bf [Local existence and uniqueness result]}
Consider an abstract initial data set $(S,g,k)$ for Einstein Vacuum Equation with $(g,P)$ such that 
\begin{eqnarray*}
\norm{R}_{L^2(S)} + \norm{\nabla R}_{L^2(S)}+ \norm{\nabla^2 R}_{L^2(S)} &\leq& \al \\
\norm{P}_{L^4(S)}+\norm{\nabla P}_{L^2(S)}+\norm{\nabla^2P}_{L^2(S)}+ \norm{\nabla^3 P}_{L^2(S)}&\leq& \al
\end{eqnarray*} 

then there exists a MGHVE \footnote{ Globally Hyperbolic Vaccum Extension}  ({\bf M,g}) with Cauchy data $(S,g,k)$ unique up to a diffeomorphism such that the future  temporal extent of any $p \in S $ denoted by $T^+(p)$, that is the maximal length of future causal curves initiating at $p$ in {\bf M} satisfies :
\begin{equation}
T^{+}(p) \geq T(\al,r_h^{3,2}(p))=min(T^\star(\alpha),r_h^{3,2}(p)T^\star(\alpha) )>0
\end{equation}

where $r_h^{3,2}(p)$ stands for the $L^{3,2}$ harmonic radius at $p$ on $S$. (see below for the definition).
\end{thm}
\begin{thm}
Suppose moreover that a scalar density $f \in H^3(S)$ is given on $S$ satisfying $\nu^{-1}\leq f \leq \nu$ for a $\nu>1$, then {\bf M} can be foliated by a time function $t$ satisfying $\square_{\bo g} t =0$, $t=0$  and $(-{\bo g}({\bo D}t,{\bo D}t))^{-\demi}=n_0=f\sqrt{det(g_0)}$ on $S$ in a neigborhood of $S$ in ${\bo M}$.
\end{thm}

{\it Remarks :}

\begin{itemize} 
\item contrary to standard local-in-time theorems, this version give an explicit bound on "how big" the MGHVE without compacity  or asymptotic flatness hypothesis . As there is no natural time in general relativity, the most natural gauge-invariant way to measure the "temporal size" is the one chosen here.
\item the gauge hypothesis of the second theorem could have been replaced by other gauge choices. It says roughly that given some datas, we know that a  neighborhood of $S$ in {\bf M} can be foliated.
\item the  important aspect of relation (38) is that it enables us not only to compare the "time of definition" of two solutions of Einstein vacuum equations corresponding to data sets given on the same manifold but also corresponding to two general abstract initial data sets $(S,g,k)$ and $(S',g',k')$. This "time of definition" depends on analytic and geometric quantities.
\end{itemize}

\begin{cor} 
If we suppose moreover that $S$ is complete and has radius of injectivity bounded by below by $i_0>0$ then for all $p \in S$, $r_h^{3,2}(p)\geq r(\norm{R}_{H^l(S)},i_0)>0$, which implies a uniform lower bound on the future temporal extent for the corresponding MGVHE. Instead of a lower bound on the radius of injectivity, we can require that the volume radius at scales less than one is uniformly bounded by below on $S$.

\end{cor}
 
\begin{cor}
Suppose that $S$ is open and that there exists $\al>0$, $D>0$, $d>0$ and $\eta>0$ and $p \in M$ such that :

\begin{itemize}
\item $d_g(p,\partial M) \geq d$
\item $vol(B(p,s)) \geq \eta s^3$ for all $s \leq d/2$
\item $\norm{\nabla^j R}_{L^2(M)}\leq Q$ for $0 \leq \abs{j}\leq l$
\end{itemize}
then $r_h^{2+l,2}(p) \geq r(\al,d,\eta) >0$, thus :
\begin{equation}
T^{+}(p)\geq T(Q,d,\eta)>0
\end{equation}

\end{cor}
 
The proof will require some results of Cheeger-Gromov theory.

\subsection{Some results from Cheeger-Gromov theory :}

Let us begin by the defining the harmonic $L^{2+l,2}$ radius of a point $p$ of a  Riemannian manifold $(M,g)$  denoted by $r^{2+l}_h(p)$ as the largest radius of a geodesic ball about $p$ on which there exists an harmonic chart in which :
\begin{eqnarray*}
\demi \delta_{ij} \leq g_{ij} &\leq& 2  \delta_{ij} \\
r_h(p)^{|j|-3/2}\norm{\partial^j g_s}_{L^2(B(p,r_h(p))} &\leq& 2 \text{ for every multiindex $ j, |j| \leq2+l$}.
\end{eqnarray*}
the first relation will be written below : $Spectrum(g)\in [1/2,2]$.

Let us now state the theorem :

\begin{thm}
Let $Q>0$,  $d>0$ and $\eta>0$ given, then there exists $r(Q,d,\eta)>0$ such that for any 3D-Riemannian manifold $(S,g)$ and $p \in M$ such that :

\begin{itemize}
\item $d_g(p,\partial M) \geq d$
\item $vol(B(p,s)) \geq \eta s^3$ for all $s \leq d/2$\item $\norm{\nabla^j R}_{L^2(M)}\leq Q$ for $0 \leq \abs{j}\leq l$
\end{itemize}
then $r_h^{2+l,2}(p) \geq r(Q,d,\eta) >0$

\end{thm}

\begin{proof}
It is a version a Cheeger-Gromov control on the harmonic radius of a manifold (see \cite{pe}).
\end{proof}

We remark that if denote by $r(Q,d,\eta)$ the largest $r$ satisfying the previous theorem, we have the following monotonicity properties :
\begin{itemize}
\item for any $(d,\eta)$, $r(.,d,\eta)$ is decreasing
\item for any $(Q,d)$, $r(Q,d,.)$ is increasing
\end{itemize}

Moreover we have that thee exists  $\mu(Q,\eta)>0$ such that  :
\begin{equation*}
r(Q,d,\eta) \geq min(\mu(Q,\eta),d\mu(Q,\eta))
\end{equation*}
\begin{proof}Fix $\lambda>0$, consider $(S,\lambda^2 g)$ and $p \in S$, denote by ${r}_h^{l+2,2}(p,\lambda)$ its $(l+2,2)$ harmonic radius on $(S,\lambda^2 g)$ and $r_h^{l+2,2}(p)$ its $(l+2,2)$ harmonic radius on $(S, g)$, then by scaling property of the harmonic radius, $r_h^{l+2,2}(p,\lambda)=\lambda r_h^{l+2,2}(p)$. But $r_h^{l+2,2}(p,\lambda)\geq r(Q(S,\lambda^2g),\lambda d,\eta)\geq  r(max(\lambda^{-1/2},\lambda^{-5/2})Q,\lambda d,\eta)$. Thus :
\begin{equation*}
r_h^{l+2,2}(p)\geq d r(max(d^{1/2},d^{5/2})Q,1,\eta)
\end{equation*}
denoting  $\mu(Q,r)=r(Q,1,\eta)$, for $d\leq1$, we have $r_h^{l+2,2}(p)\geq d \mu(Q,\eta)$. If $d\geq 1$, consider $(B(p,1),g)\subset (S,g)$, and $r''(p)$ the $(l+2,2)$ harmonic radius of $p$ on $(B(p,1),g)$, we have $r_h^{l+2,2}(p)\geq r''(p)\geq r(Q(B(p,1),g),1,\eta)\geq  r(Q,1,\eta)$.
\end{proof}

Thus (3.3.2) is a corollary of (3.2.1). We will now prove theorem (3.2.1).

\subsection{Proof of the local existence theorem :}

\hspace{5mm}We begin by using the harmonic coordinates to identify $B_g(p,r_h^{3,2}(p)) \subset {\bf M}$ with a $B'_g(p,r_h(p)) \subset\mathbb{R}^3$. Here $\mathbb{R}^3$ will be endowed with its standard Sobolev norms. Remark that thanks to the definition of the harmonic radius. 
\begin{eqnarray*}
\norm{g}_{H^3(B'_g(p,r_h(p)))}& \leq& C(r_h(p)) \\
\norm{k}_{H^2(B'_g(p,r_h(p)))}& \leq &C(r_h(p),\al)
\end{eqnarray*}

 We can extend smoothly on $\mathbb{R}^3$, $g$ and $k$ thus  constructing $(\tilde{g},\tilde{k})$ satisfying $\norm{\tilde{g}-e}_{H^3(\mathbb{R}^3)}\leq 2 C$, $\norm{\tilde{k}}_{H^{2}(\mathbb{R}^3)}\leq 2C$, where $e$ stands for the euclidian norm of $\mathbb{R}^3$. We can also suppose that $(\tilde{g},\tilde{k})=(e,0)$ outside  $B'_g(p,2r_h(p))$ and $Spectrum(\tilde{g})\subset [2/5,5/2]$. We do not require of course $(\tilde{g},\tilde{k})$ to satisfy the constraint equations outside $B'_g(p,r_h(p))$.

We then solve the hyperbolic system $\mathcal{S}$ with data $(\tilde{g},\tilde{k})$ for  a $ \nu^{-1}\leq n_0 \leq \nu$ given in $H^3$. Using the Cauchy stability theorem for this reduced system, we obtain the existence of a time $T_1(C,n_0)>0$ such that the solution of  $\mathcal{S}$ with initial data $(\tilde{g}, \tilde{k})$ exists on $\mathbb{R}^3\times [-T_1,T_1]$ and $(g(x,t)-e(x)) \in \mathcal{C}([0,T];H^{3}(\mathbb{R}^{3})) \cap \mathcal{C}^{1}([0,T];H^{2}(\mathbb{R}^{3}))$ depends continuously on the initial data. The propagation of the gauge condition ensures that on the domain of dependance of $B'_g(p,r_h(p))$ the previously constructed solution is a solution of the Einstein vacuum equations.

 Using the Cauchy stability theorem for the reduced system as  well as Sobolev embedding theorem, there exists $T_2(r_h(p),\al,n_0)>0$, such that for every time $ t \in [-T_2,T_2]$ and $x \in B'_g(p,r_h(p))$ :
 
 \begin{eqnarray}
 n(t,x) \in [\demi\nu^{-1},2\nu] \\
 \text{Spectrum }\tilde{g}(t,x) \subset [1/4,4]
 \end{eqnarray}

Let us now consider $p_t:=(p,t)$ and prove that there exists $T_3$ such that for $\abs{t} < T_3(r_h(p),\al)\leq T_2(r_h(p),\al)$, any causal  past-directed curve initiating at $p_t$ crosses $\mathbb{R}^3 \times \{ 0\}$ in $B'(p,r_h)$ that is $p_t$ belongs to the domain of dependance of $B'(p,r_h)$. 

To prove this property, let us parametrize a causal curve $\gamma(s)=(x(s),t-s)$ with $s \in [0,t]$ and $x(0)=p$. Using the fact that $\gamma$ is causal :
 \begin{equation*}
 \tilde{g}_{ij}(x,t)\dot{x}_i \dot{x}_j \leq 2\nu\\
\end{equation*}
but :

 \begin{equation*}
 \frac{1}{10}\tilde{g}_{ij}(x,0)\dot{x}_i \dot{x}_j \leq \frac{1}{4} \Sigma_i \abs{\dot{x}_i}^2 \leq \tilde{g}_{ij}(x,t)\dot{x}_i \dot{x}_j \leq 2\nu\\
\end{equation*}

integrating in $s$ we deduce $\gamma(t) \in B'(p,20\nu t)$. Thus taking :
\begin{equation}
T_3(r_h(p),\al,\nu)=min(T_2(r_h(p),\al,\nu),\frac{r_h(p)}{30 \nu })
\end{equation}
we deduce that for $\abs{t} < T_3(\al,r_h(p))\leq T_2(\al,r_h(p))$, any causal curve past-directed initiating in $p_t$ crosses $\mathbb{R}^3 \times \{ 0\}$ in $B'(p,r_h)$ that is $p_t$ belongs to the domain of dependance of $B'(p,r_h)$. We thus have :

\begin{equation}
T^{+}(p)\geq inf(n) T_3(r_h(p),\al,\nu)\geq T'_3(r_h(p),\al,\nu)>0.
\end{equation} 
 
  Once an abstract initial data set is given $(S,g,k)$, we can fix an arbitrary lapse and $(42)$ gives an explicit lower bound for the temporal extent of every point $p \in S$ depending only on the initial data set, the given lapse, and the harmonic radius at $p$. We patch the local solutions to obtain a GHVE $({\bf M,g})$ of $(S,g,k)$ foliated by an harmonic time function.
Now if we consider a initial data set $(S,g,k)$ for Einstein vacuum extension satisfying the conditions of theorem (3.2.1), the corresponding MGVHE is such that a neighborhood of $S$ is diffeomorphic to the previously constructed GHVE, theorems (3.2.1) and (3.3.2) follow.

\vspace{3mm}
Let us now consider $T^\star(\al,r)$ the largest time for which theorem 3.2.1 holds, we have the following monotonicity properties :
\begin{itemize}
\item $T^\star(.,r)$ is decreasing
\item  $T^\star(\al,.)$ is increasing
\end{itemize}
moreover there exists $T^\star_1(\al)>0$ s.t. :
\begin{equation*}
T^\star(\al,r)\geq min(T^\star_1(\al),r_h^{3,2}(p)T^\star_1(\al))
\end{equation*}
\begin{proof}
The MGVHE associated with $(S,\lambda^2g,\lambda k)$ for a fixed $\lambda>0$ is $({\bo M}, \lambda^2{\bo g})$ where $({\bo M}, {\bo g})$ stands for the MGVHE of $(S,g,k)$. Consider $p\in S$, its future temporal extent in $({\bo M}, \lambda^2{\bo g})$, $T^{+}_{\lambda}(p)=\lambda T^{+}(p)$, we thus deduce that :
\begin{equation*}
\forall \lambda>0, T^\star(\al,r)\geq \lambda^{-1}T^\star(max(\lambda^{-1/2},\lambda^{-5/2})\al,\lambda r)
\end{equation*}

Now let us define $T^\star_{1}(\al)=T^\star(\al,1)$, if $r_h^{3,2}\leq1$ we have :
\begin{equation*}
T^\star(\al,r)\geq r_h^{3,2}T^\star((r_h^{3,2})^{1/2}\al,1)\geq r_h^{3,2}T^\star_1(\al)
\end{equation*}
 If $r_h^{3,2}\geq1$, consider $T'(p)$ the future temporal extent of $p$ in the MGHVE of $(B(p,1),g,k)$, then $T^{+}(p)\geq T'(p)\geq T^\star(\al', 1)$ where $\al'<\al$ stands for the corresponding norm on $B(p,1)$, using the monotonicity of $T^\star$, $T^{+}(p)\geq T^\star_1(\al)$.
\end{proof}

Now consider an initial data set $(S,g,k)$ such that $diam(S)<+\infty$ and $(S,g,k)$ satisfies the conditions of corollary (3.2.2), we consider a sequence of points $p_n \rightarrow \partial S$, we have that there exists $T(Q,\nu)>0$, $T^{+}(p_n)\geq dist(p_n, \partial S)T(Q,\nu)$. We know that for non-complete initial data set, the temporal extent of the corresponding MGVHE typically tends to zero while approaching the boundary of $S$, the previous control shows that rate of decrease is at worst linear in the distance to the boundary. This phenomenon can be checked for instance in the case of the simplest example of initial data set for corollary 3.2.2  $(S,e,0)$, where $S$ is the open ball of radius $1$ in $ \mathbb{R}^3$, $e$ the euclidian metric of $ \mathbb{R}^3$. Its MGVHE is $({\bo C},{\bo m})$ where ${\bo C}=\{(x,y,z,t) \in \mathbb{R}^4 | x^2+y^2 +z^2 <(1-t)^2 \}$ and $m$ the Minkovsky metric. Let $p(x,y,z) \in S$, then $T^{+}(p)=\sqrt{1- x^2-y^2-z^2}\geq dist(p,\partial S)$.

 \subsection{Application to the local structure of Einstein vacuum spacetimes :}
 
\hspace{5mm}
Consider a spacetime {\bf (M,g)} supposed to be strongly causal and an achronal spacelike hypersurface  $\Omega_0$ satisfying the condition of local-in-time theorem and {\bf (G)}. The existence of such a local slice in the neighborhood of a point of a spacetime is not a strong requirement. Given a regular spacetime {\bf (M,g)}, a point $p$ and a future-directed normalized vector $\bo{T} \in T_p(\bo M)$, it's always possible to construct locally an achronal spacelike hypersurface  $\Omega_p$ such that $p \in\Omega_p$ and  {\bf T} is the normal to the hypersurface at $p$. It sufficient to remark that in a local coordinate chart $(t,x_1,x_2,x_3)$, we can construct this hypersurface as the graph of a function $t=t(x_1,x_2,x_3)$.
Let us consider a domain $\mathcal{D}$ of $\bf{M}$ globally hyperbolic with respect to $\Omega$,  foliated by a harmonic time-function with  $\Omega$ corresponding to $t=0$ and such that its boundary consists of three parts $\Omega_0$, $\Omega_{t_0}$ a spacelike hypersurface level set of $t$ and a null lateral boundary $\mathcal{H}$. We suppose the boundary of $\mathcal{D}$ to be Lipschitz regular.
\vspace{3mm}

 For a subset of $\Omega'$ of $\Omega$, we denote by $\Omega'_t$ the subset of {\bf (M,g)} obtained by following the integral lines of $\frac{\partial}{\partial t}=n{\bf T}$ during time $t$ and by $F_t : \Omega' \rightarrow \Omega'_t$ the corresponding map.

We suppose in this part that the initial data satisfies :

\begin{equation*}
a) \int_{\Omega} \abs{{\bf R}}^2(x)d\mu(x) \ls \Delta_0\\
\end{equation*}

\begin{thm}{\bf [Local theorem] }
Suppose that there exists  $0<\Delta_1<\infty$ s.t. :
\begin{eqnarray*}
&b)&\norm{^{\bf (T)}{\pi}}_{L^1_t(L^\infty(\Omega_t)} \ls \Delta_1 \\
&c)& \text{there exists a non-empty open set $\Omega' \subset \Omega_0$ s.t. $F_t(\Omega')\subset \Omega_t$, $\forall t\in[0,t_0[$}
\end{eqnarray*}
 
then the spacetime together with the harmonic time function can be extended  beyond time $t_0$ in the neighborhood of any point of $\Omega_{t_0}$. 

\end{thm}
 
Let us make some remarks :
\begin{itemize}
\item $c)$ is a non-crushing condition. 
\item The other conditions are similar to the global criterion in the form proved by Q.Wang (see \cite{qw2} and \cite{qw3}), but localized into $\mathcal{D}$.
\item the choice of $\mathcal{D}$ as an appropriate set to localize the breakdown criteria is natural in the sense that if we consider a point $p \in \Dep$ and $\mathcal{J}^-(p)$ its causal past than $\mathcal{J}^-(p)$ is divided into $\mathcal{J}^-(p) \cap \Dep$ and $\mathcal{J}^-(p) \cap \mathcal{J}^-(\Omega)$ thus in particular arguments on the causal structure of  $\mathcal{N}^-(p)$, energy arguments which are local remain valid as we will see later on. 
\item using (24), we deduce that there exits $\Delta_3>1$ such that $\Delta_3^{-1}<n<\Delta_3^{-1}$ on $\mathcal{D}$.
\item condition $b)$ can be replaced by a uniform bound on the Riemann tensor which is a formally stronger assumption.
\end{itemize}

\subsection{Main Lemmas :}
Let us denote by $\Omega'_t=F^t(\Omega')$. The previous theorem will be proved using the following lemmas :

\begin{lem} {\bf [Geometric control]}
The family $\Omega'_t$ satisfy the following properties : there exists $s_0<1$, $C>0$ and $D>0$ such that for all  $t\in [0,t_0[$ and all $x \in \Omega'_t$ such that :
\begin{eqnarray*}
&a)&  \text{the metric } g_t \text{ satisfy the geometric assumption {\bf (G)} uniformly for } \\
&&
t \in [0,t_1[ \\
&b)&  vol_{g_i}B(x,s)\rs C s^3 \text{ for } s\ls s_0 dist(x,\partial \Omega'_t) \\
\end{eqnarray*}
\end{lem}
\begin{lem}{\bf [Energy control]}
The $\Omega'_t$ satisfy the following condition for all  $t\in [0,t_1[$  :
\begin{eqnarray*}
&&(R_t,P_t) \text{ is uniformly bounded in the } H^{2}(\Omega'_{t})\times H^{3}(\Omega'_{t})\text{ topology by a constant}\\
&&C(\Delta_i,t_0)\norm{{\bf R}}_{H^2(\Omega)} .\\
\end{eqnarray*}
\end{lem}
\subsubsection{Proof of lemma (3.6.1) :}
\begin{proof}
Let us begin with the Lemma 1 :

\begin{equation}
\partial_t g_{ij}=-2nk_{ij}
\end{equation}

Consider a point $p$ of $\mathcal{D}$, following the integral lines of $-\partial_t$ and using the global hyperbolicity of $\mathcal{D}$, the integral line intersects $\Omega$ at a point $p_0$. Denoting by $C$ the constant of the hypothesis {\bf (G)}, we get thanks to the use Gronwall's lemma that {\bf (G)} is satisfied uniformly on the slides $\Omega'_t$ for a constant $C'=Cexp(\int_0^{t_1}\norm{nk}_{L^\infty(\Omega^m_t)}dt )$ that is as matrices :
\begin{equation*}
C'^{-1}\delta \leq g_t \leq C' \delta
\end{equation*}
 We thus obtain that $F_t : \Omega' \rightarrow \Omega'_t$ is a quasi-isometry with constant not depending on $t \in [0,t_0[$, more precisely, there exists $J>1$ such that for every $(p,q) \in \Omega'$ and every $t \in [0,t_0[$ :

\begin{equation*}
J^{-1}d_{g_t}(F^t(p),F^t(q)) \leq d_{g_0}(p, q)  \leq J d_{g_t}(F^t(p),F^t(q))
\end{equation*}
To control  the volume radius at fixed scales, we first remark that there exists $C''$ such that for all $x$, for all $t$, $x \in \Omega'_t$,
\begin{equation*}
B_{g_t}(x,a)\supset B_{e}(x,C''a)
\end{equation*} 
moreover $\partial_t log(detg)=-2ntr(k)$, thus :
\begin{equation*}
Vol(B_{g_t}(x,a)) \geq \int_{ B_{e}(x,C''a)}\sqrt{det(g_t)}dx \geq J(\Delta_1,\Delta_2,C'')a^3
\end{equation*}
\end{proof}
\subsubsection{Proof of lemma (3.6.2) :}

\hspace{5mm}
We only sketch in this part the proof of the energy lemma as its proof is essentially the same as the one given in \cite{krb}. We insist only on the differences related to the use of the time-harmonic transversal gauge and the local energy estimates. We refer the reader to this paper for the details of the proof.

 Let us consider a point $p \in \mathcal{D}$, the local geometry near the vertex of the past null cone $\mathcal{N}^{-}(p)$ can be controlled exactly as in the global in space case. We can prove exactly as in the global case a bound from below for the null radius of injectivity of  $\mathcal{N}^{-}(p)$ assuming only a bound of a reduced $L^2$ flux of the curvature along the null cones. implying a pointwise control on the Riemann curvature tensor as well as a $H^4 \times H^3$ control on $(g,k)$ as in \cite{krb} and \cite{dp}. This requires sophisticated techniques which require to consider from the one hand the Bel-Robinson tensor and from the other hand the non-linear wave equation satisfied by {\bf R} and its covariant derivative together with control on the causal geometry of null cones (see \cite{krb}, \cite{krk}, \cite{krc}, \cite{krg}).  This control depends only on the constants $\Delta_i$ and $t_0$. We refer the reader to the previously cited article for details. To be more precise :
\begin{thm}
Let $\bo{R}$ be the Riemann curvature tensor of a solution of the EVE and satisfying the hypothesis $a), b) $ and $c)$ of the local theorem, then for all t in $[0,t_0[$ and $p\in \mathcal{D}$ :

\begin{eqnarray}
&a)& \int_{\Omega_{t}} \abs{\bo{R}}^2(x)d\mu(x) \ls C_1(\Delta_i,t_0)\\
&b)& \int_{\mathcal{N}^{-}(p)\cap \mathcal{J}^{+}(\Omega)} \bo{Q[R]}(\bo{T},\bo{T},\bo{T},L) \ls C_2(\Delta_i,t_0)
 \end{eqnarray}

where $\mathcal{N}^{-}(p)$ stands for the past null cone initiating at the point p that is the boundary of the causal past of $p$ and the integral is taken in time between $0$ and $t_{1}$ and $L$ for the null geodesic generator of $\mathcal{N}^{-}(p)$, that is the vectorfield on $\mathcal{N}^{-}(p)$ satisfying :
\begin{equation}
 \bo{D}_{\bo{L}} \bo{L}=0 \hs <\bo{L},\bo{L}>=0
\end{equation}
and the normalization condition at $p$ :
$$<\bo{L},\bo{T}>(p)=1$$
\end{thm}
 
Moreover for higher order derivatives, we use the wave equation for {\bf R} :
\begin{equation*}
\square_{\bo g}{\bf R} = {\bf R}\star{\bf R}
\end{equation*} 
 
where $\star$ is a bilinear symmetric operator on the four-times covariant tensor. Its covariant derivative satisfies  :
\begin{equation*}
\square_{\bf g} {\bf D}{\bf R} \approx {\bf DR}\cdot {\bf R}
\end{equation*}
 
 using  local energy estimates for  wave equations :
 
\begin{thm}
Under the hypothesis of the local theorem, there exist a constant $C=C(t_0,\Delta_i)$ such that for all $t$ in the time slab $[0,t_1[$ :

$$\normf{\bo{DR}}_{L^2(\Omega_{t})}^2\leqslant C \big(\normf{\bo{DR}}_{L^2(\Omega_{0})}^2 +  \int_0^t\normf{\bo{R}}_{L^4(\Omega_{s})}^4nds\big)$$
\end{thm}

\begin{thm}
With the hypothesis of the local theorem, there exist a constant $C'=C'(t_0,\Delta_i)$ such that for all $t$ in the time slab $[0,t_0[$ :

$$\normf{\bo{D^2R}}_{L^2(\Omega_{t})}^2\leqslant C' \big(\normf{\bo{D^2R}}_{L^2(\Omega_{0})}^2 +  \int_0^t\normf{\bo{R}}_{L^\infty(\Omega_{s})}^2\normf{\bo{DR}}_{L^2(\Omega_{s})}^2nds\big)$$
\end{thm}

As in the global case, our problem is thus to control the pointwise norm of ${\bf R}$. Such a control is obtained by proving that there exists a control on the past null radius of injectivity of $p$, that is there exists $i_0>0$ such that $\forall p \in \mathcal{D}$, $i^{-}(p)\geq min(d(p,\Omega), i_0)$. Such a property as well as the parametrix of \cite{krk} which are local can be used as in \cite{krb} to close the estimate and obtain a uniform $H^2(\Omega'_t)$ control on {\bf R}. It remains however to derive control on $k$ and $n$ from the control on ${\bo R}$. The method used previously for a CMC or an AF gauge is from the one hand a three-dimensional Hodge system for $k$ and form the other hand the lapse equation (see \cite{krb}) :

\begin{eqnarray*}
\nabla^i k_{ij}&=&0\\
\nabla \wedge k&=&E\\
trk&=&0\text{ (AF gauge), } trk=t \text{ (CMC gauge)} \\
\Delta n &=& n\abs{k}^2 - \epsilon 
\end{eqnarray*}
where $\epsilon=0$ in the case of the AF gauge, and  $\epsilon=1$ in the case of a CMC gauge and $E$ stands for the electric part of the electric-magnetic decomposition of {\bf R}.
On a manifold without boundary we have using integration by part that if ${\bf R}\in L^2(\Omega)$, $\nabla k$ and $\nabla^2 n \in L^2(\Omega)$. Moreover, the time derivatives can be controlled using the evolution equation for $k$, obtaining $\partial_t k \in L^2(\Omega)$ and $\nabla \partial_tn\in L^2(\Omega)$.
This approach is not well-suited for local in space and time application due to the ellipticity of the estimates used.

 Instead of these techniques , we use hyperbolic estimates of theorem (3.1.1) to deduce  that :
\begin{thm}
 Consider a spacetime foliated by a time-harmonic function $t$, and such that ${\bo {D^l R}} \in L^\infty_tL^2_x(\Omega_t)$ for $\abs{l} \leq 2$ and $ ^{\bf(T)}{\pi} \in L^1_tL^\infty(\Omega_t)$, then
 \begin{eqnarray*}
 {\bf {D^l}}   ^{\bf(T)}{\pi} &\in &L^\infty_t L^2_x(\Omega_t) \text{, $\abs{l} \leq 3$} \\
 {\bf {D^l}}  R &\in& L^\infty_tL^2_x(\Omega_t) \text{, $\abs{l} \leq 2$}
 \end{eqnarray*}
 \end{thm}
   We thus avoid completely in this gauge the trace theorems from $\Sigma_t \rightarrow \mathcal{C}$ where $\mathcal{C}$ is a null cone and the associated loss of Sobolev regularity of $\demi$ and get a control on the whole flux of $^{(\bf T)}{\pi}$ on $\mathcal{C}$ and not only on the flux associated to the wave equation for $^{(4)}{k}$.   
   
   As in \cite{krb}, \cite{dp} and \cite{qw2}, we thus have to derive a $H^2(\Omega_t)$ control on ${\bo R}$. If instead of a $L^1_tL^{\infty}_x$ control on $^{\bo(T)}{\pi}$, we had only required a $L^2_tL^{\infty}_x$ control, the proof would have been similar to the one in \cite{dp}, however the stronger assumption made here requires to prove that the $L^1_tL^{\infty}_x$ bound is sufficient to control the geometry of past null cones. This require the use of the main result of \cite{qw3} whose statement requires some definitions. Note that in this paper the author gave a proof assuming a generic time foliation in comparison with \cite{krc} in which the authors studied only the case of geodesic foliation of a (troncated) null cone.
   
Let ${\bf (M,g)}$ be a smooth 3+1 Einstein vacuum spacetime foliated by $\Sigma_t$, the level hypersurfaces of a time function t with lapse function $n$. Consider an outgoing null hypersurface $\mathcal{H}= \cup_{0<t<1}S_t$ in $\bo{(M,g)}$ initiating from a point $p$, whose leaves are $S_t=\Sigma_t\cap\mathcal{H}$ and $t(p)=0$. Let us recall some definitions :
\begin{thm}{\bf [Ricci coefficients]}
The Ricci coefficients relative to an adapted null frame $(e_A,e_B,\bo{L},\bo{\Lb})$ are defined as follows :
\begin{eqnarray*}
\chi_{AB}=\langle \bo{D_A}\bo{L},e_B \rangle &\hspace{5mm}& \underline{\chi}_{AB}=\langle\bo{D_A}\bo{\Lb},e_B \rangle \\
 \xi_{A}=\frac{1}{2}\langle \bo{D_\bo{L}}\bo{L},e_A \rangle  &\hspace{5mm}& \underline{\xi}_{A}=\frac{1}{2}\langle \bo{D_{\bo{\Lb}}}\Lb,e_A \rangle\\
\h_A=\frac{1}{2}\langle \bo{D_{\bo{\Lb}}}\bo{L},e_A \rangle&\hspace{5mm}&  \underline{\h}_{A}=\frac{1}{2}\langle \bo{D_\bo{L}}\bo{\Lb},e_A\rangle \\
\omega =\frac{1}{4}\langle \bo{D_\bo{L}}\bo{L},\bo{\Lb} \rangle&\hspace{5mm}&\underline{\omega}=\frac{1}{4}\langle \bo{D_{\bo{\Lb}}}\bo{\Lb},\bo{L}\rangle \\
V_A=\frac{1}{2}\langle \bo{D_A}\bo{L},\bo{\Lb}\rangle
\end{eqnarray*}
We recall the definition of the null components of a Weyl Field $W$ relative to a null pair $(\bo{L},\bo{\Lb})$ :
\begin{eqnarray*}
&&\al(W)(X,Y)=W(X,\bo{L},Y,\bo{L})\\
&&\bet(W)(X)=\frac{1}{2}W(X,\bo{L},\bo{\Lb},\bo{L})\\
&&\rho(W)=\frac{1}{4}W(\bo{\Lb},\bo{L},\bo{\Lb},\bo{L})\\
&&\sigma(W)=\frac{1}{4}{}^{\star}W(\bo{\Lb},\bo{L},\bo{\Lb},\bo{L})\\
&&\underline{\al}(W)(X,Y)=W(X,\bo{\Lb},Y,\bo{\Lb})\\
&&\underline{\bet}(W)(X)=W(X,\bo{\Lb},\bo{\Lb},\bo{\Lb})
\end{eqnarray*}
denote moreover by :
\begin{eqnarray*}
a^{-1}&=&-<\bo{L}, \bo{T}> \\
\mu &=& -\demi \bo{D}_{\bo{\Lb}} tr \chi +\frac{a^2}{4}(tr(\chi))^2 - \underline{\omega} tr(\chi) \\
4 \pi r^2(t)&=& \int_{S_t} 1 d\mu_{S_t}
\end{eqnarray*}
\end{thm}

moreover $\mu$ stands for the mass function, $a$ for the null lapse, $r$ for the radius of the level sets $\{ t=cte\}$ foliation of a null cone. For the definition of the appropriate functional spaces $\mathcal{P}$ and $\mathcal{B}$ see \cite{krc}.

More precisely :
\begin{thm} {\bf [Control on the null Ricci coefficients of null cones, Q.Wang, \cite{qw3}]}

Using the previously introduced notations, assume $C^{-1}<n<C$ on $\mathcal{H}$ for $C>0$ and :

\begin{equation*}
\mathcal{R}(\mathcal{H})+ \mathcal{F}(\pi)\leq \mathcal{R}_0, \text{ on $\mathcal{H}$}
\end{equation*}
with $ \mathcal{R}_0$ sufficiently small. Then the following estimates hold true :

\begin{eqnarray*}
\norm{tr\chi-\frac{2}{s}}_{L^{\infty}(\mathcal{H})}&\lesssim& \mathcal{R}^2_0 \\
 \abs{a-1}&\leq &\demi \\
\norm{\int_0^1\abs{\hat{\chi}}^2nadt}_{L^{\infty}_\omega}+ \norm{\int_0^1\abs{\zeta}^2nadt}_{L^{\infty}_\omega}&\lesssim& \mathcal{R}^2_0 \\
 \norm{\int_0^1\abs{\nu}^2nadt}_{L^{\infty}_\omega}+ \norm{\int_0^1\abs{\underline{\zeta}}^2nadt}_{L^{\infty}_\omega}&\lesssim &\mathcal{R}^2_0 \\
\norm{\nablab tr\chi}_{\mathcal{P}_0}+ \norm{\mu}_{\mathcal{P}_0}+ \norm{\nablab tr\chi}_{L^2{\mathcal{H}}}+  \norm{\mu}_{L^2{\mathcal{H}}} &\lesssim &\mathcal{R}_0 \\
\mathcal{N}_1(\hat{\chi})+ \mathcal{N}_1(\zeta)+  \mathcal{N}_1(tr\chi-\frac{2}{r})+\mathcal{N}_1(tr\chi-(an)^{-1}\overline{antr\chi})&\lesssim& \mathcal{R}_0 \\
\norm{tr\chi-\frac{2}{r}}_{L^2_tL^{\infty}_\omega}+\norm{tr\chi-(an)^{-1}\overline{antr\chi}}_{L^2_tL^{\infty}_\omega}&\lesssim &\mathcal{R}_0 \\
 \norm{\sup_{t\leq1}\abs{r^{3/2}\nablab tr \chi}}_{L^2_\omega}+\norm{\sup_{t\leq1}\abs{r^{3/2}\mu}}_{L^2_\omega}+ \norm{r^{1/2}\nablab tr \chi}_{\mathcal{B}_0}+ \norm{r^{1/2}\mu}_{\mathcal{B}_0}&\lesssim& \mathcal{R}_0 
\end{eqnarray*}
\end{thm}
{\it Remark :} We have here a slighlty better control on the flux of order one of $^{(\bo T)}{\pi}$ due to the fact that we have obtained a wave equation on $^{(\bo T)}{\pi}$ and not only on $k_{\alpha \beta}={\bo P}_{\alpha \nu}{\bo P}_{\beta \mu} {^{(\bo T)}{\pi}^{\mu \nu}}$ like in \cite{qw2} and \cite{qw3}.

The control of the geometry of the null cones require also to use the following theorem :

\begin{thm} {\bf [Control on the radius of injectivity of null hypersurfaces, Klainerman and Rodnianski, \cite{kri}]}

Let ${\bf (M,g)}$ be a smooth 3+1 Einstein vacuum spacetime foliated by $\Sigma_t$, the level hypersurfaces of a time function t with lapse function $n$. Suppose that $\Sigma_0$ satisfies ${\bf (G)}$. Consider an outgoing null hypersurface $\mathcal{H}= \cup_{0<t<1}S_t$ in ${\bo (M,g)}$ initiating from a point $p$, whose leaves are $S_t=\Sigma_t\cap\mathcal{H}$ and $t(p)=0$. Assume $C^{-1}<n<C$ on $\mathcal{H}$ for $C>0$, and $\norm{^{\bo (T)}{\pi}}_{L^1_tL^{\infty}_x}\leq D$. Suppose also that $\norm{{\bo R}}^2_{L^2(\Sigma_{0})}+ \norm{k}^4_{L^4(\Sigma_{0})}\leq E<\infty$, then there exists $i_0(C,D,E)>0$ s.t. for all $p \in \mathcal{D}^{+}(p)$ :

\begin{equation}
i^{-}(p,t)\geq min(i_0,d_{\bo g}(p,\Omega))
\end{equation}

\end{thm}
{\it Remark :} The reader will note that the two theorems come together : some control on the null Ricci coefficients are needed to control the radius of injectivity, in particular $tr\chi$ (see \cite{kri}). From another hand, the bootstrap argument of \cite{krc}, \cite{dp} or \cite{qw3} requires to work on a smooth part of a past null cone that is on a $\mathcal{N}^{-}(p,\delta)$ which requires some uniform control on the radius of injectivity of these null cones.

\vspace{3 mm}
Finally if the null cones satisfy uniformly the control on the radius of conjugacy and on the null Ricci coefficients, one can obtain using a parametrix of Kirchoff-Sobolev type (see \cite{krk}) following the orginal strategy that ${\bf R}$, ${\bf D}{\bf R}$ and ${\bf D^2}{\bf R}$ are bounded in $L^2(\Omega_s)$ by a constant depending only on the initial data, the $\Delta_i$ and $t_0$.

\vspace{3 mm}

Thus the proof of the energy estimates follows the original strategy of \cite{krb } to derive uniform $L^2$ control on ${\bf R}$ and its covariant derivatives with two slight modifications. From the one hand, we use purely hyperbolic estimates to control energy estimates on $^{(\bo T)} {\pi}$ from the energy estimates on ${\bf R}$ using the structure of the reduced system in a time harmonic transversal gauge and from the other hand we use local energy estimates instead of estimates in a strip.

\subsection{Proof of the theorem :}

\begin{proof}  Consider a point $p \in \Omega'$, denote by $p_t=F_{t}(p)$. Using the equivalence of the uniform ellipticity of the metrics $g_t$ given by energy lemma  we can find $r_p>0$ and $C_p>0$ such that the volume radius of $p_t$ at scales $ \leq r_p$ on $\Omega'_t$  is bounded by below by $C_p$ independently of $t$. Using the energy lemma as well as the geometric lemma, we can apply the local in time existence theorem to $(\Omega_{t},g_{t},k_{t})$ and points $p_t$ to deduce that there exists a Cauchy development of $(\Omega_{t},g_{t},k_{t})$ for proper time $T(p)>0$ (not depending of t) to the future of $p_t \in \Omega_{t}$ . From another hand, the conditions on the theorem imply that the maximal proper time between the slices $\Omega_i$ and $\Omega_j$ for $0\leq t_i < t_j<t_0$  denoted by $T_i^j$ satisfies :
\begin{equation}
 T_i^j \ls \Delta_1(t_j-t_i)
\end{equation}
Indeed using :
\begin{equation}
{\bf g}= -n^2 dt^2 +g_{ij}dx^i dx^j
\end{equation}
 a uniform bound on the lapse implies directly (49).

We have as a consequence that for $t$ large enough the maximal Cauchy development of $(\Omega_{t},g_{t},k_{t})$  extends as a globally hyperbolic spacetime to the future of any $F^t(\Omega')$ , $0 \leq t<t' < t_0$ in the neighborhood of $p_t$, which concludes the proof. 
\end{proof}

{\bf Acknowledgement :} I would like to thank Professor Planchon for useful discussions. Professor Klainerman suggested the subject to me and the interest of using a time-harmonic gauge to localize the criterion. I am very grateful to him for the discussions we had on the subject and the time he spent in discussions with me during the months i spent in Princeton.

\bibliographystyle{plain}
\bibliography{local_submitted3}

\end{document}